\documentclass[11pt]{article}
\usepackage{array,multirow}
\usepackage{wrapfig}
\usepackage{times}
\usepackage{latexsym}
\usepackage{euscript}
\usepackage{amstext}
\usepackage{graphicx}
\usepackage{color}
\usepackage{url,hyperref}
\usepackage{verbatim}
\usepackage{xspace}
\usepackage{framed}
\usepackage[ruled]{algorithm2e}

\usepackage{algpseudocode}
\usepackage{amsmath,amsfonts,amssymb,bbm,amsthm} 
\usepackage{tikz, comment, subfig, enumerate}

\newif\iflong

\longtrue

\newtheorem{lemma}{Lemma}[section]
\newtheorem{theorem}[lemma]{Theorem}

\newtheorem{definition}[lemma]{Definition}
\newtheorem{corollary}[lemma]{Corollary}

\newtheorem{claim}[lemma]{Claim}

\newtheorem{observation}{Observation}

\begin{document}

\title{The Core of the Participatory Budgeting Problem}
\author{Brandon Fain\thanks{Department of Computer Science, Duke University, 308 Research Drive, Durham, NC 27708. {\tt btfain@cs.duke.edu}. Supported by NSF grants CCF-1637397 and IIS-1447554.}  \and Ashish Goel\thanks{Supported by the Army Research Office Grant No. 116388, the Office of Naval Research Grant No. 11904718, by NSF grant CCF-1637418, and by the Stanford Cyber Initiative. Author's Address: Management Science and Engineering Department, Stanford University, Stanford CA 94305. Email: {\tt ashishg@stanford.edu}} \and
Kamesh Munagala\thanks{Department of Computer Science, Duke University, Durham NC 27708-0129. {\tt kamesh@cs.duke.edu}. Supported by NSF  grants CCF-1408784, CCF-1637397, and IIS-1447554.}}

\date{}
\maketitle

\begin{abstract}
In \textit{participatory budgeting}, communities collectively decide on the allocation of public tax dollars for local public projects. In this work, we consider the question of fairly aggregating the preferences of community members to determine an allocation of funds to projects. This problem is different from standard fair resource allocation because of \textit{public goods}: The allocated goods benefit all users simultaneously. Fairness is crucial in participatory decision making, since generating equitable outcomes is an important goal of these processes. We argue that the classic game theoretic notion of core captures fairness in the setting.  To compute the core, we first develop a novel characterization of a public goods market equilibrium called the \textit{Lindahl equilibrium}, which is always a core solution.  We then provide the first (to our knowledge) polynomial time algorithm for computing such an equilibrium for a broad set of utility functions; our algorithm also generalizes (in a non-trivial way) the well-known concept of  proportional fairness. We use our theoretical insights to perform experiments on real participatory budgeting voting data. We empirically show that the core can be efficiently computed for utility functions that naturally model our practical setting, and examine the relation of the core with the familiar welfare objective. Finally, we address concerns of incentives and mechanism design by developing a randomized approximately dominant-strategy truthful mechanism building on the exponential mechanism from differential privacy.
\end{abstract}
%%% Local Variables: 
%%% mode: latex
%%% TeX-master: "main"
%%% End: 

\section{Introduction}
%Transparency and citizen involvement are fundamental goals for a healthy democracy. In recent years, \textit{participatory budgeting}, in which communities collectively decide on the allocation of public tax dollars for local public projects, has emerged as an additional means of accomplishing these goals~\cite{PBP}. However, implementing participatory budgeting requires careful consideration of how to aggregate the preferences of community members into an actionable project funding plan. Researchers have viewed this challenge as a social choice problem~\cite{goelSocialChoice,goel2015knapsack}; in this work, we model participatory budgeting as a fair resource allocation problem. We note that this problem is different from standard fair resource allocation because of \textit{public goods}: The allocated goods benefit all users simultaneously. We model this problem as a central body fairly allocating public goods according to preferences (or utility functions) reported by the community members (or users), subject to a budget constraint. It is important to note that in participatory democracy processes, equitable and fair outcomes are an important systemic goal.

Transparency and citizen involvement are fundamental goals for a healthy democracy. Participatory Budgeting (PB)~\cite{pb1,pb2} is a process by which a municipal organization (eg. a city or a district) puts a small amount of its budget to direct vote by its residents. PB is growing in popularity, with over 30 such elections conducted in 2015.  Implementing participatory budgeting requires careful consideration of how to aggregate the preferences of community members into an actionable project funding plan. In this work, we model participatory budgeting as a fair resource allocation problem. We note that this problem is different from standard fair resource allocation because of \textit{public goods}: The allocated goods benefit all users simultaneously. We model this problem as a central body fairly allocating public goods according to preferences reported by the community members (or users), subject to a budget constraint. It is important to note that in participatory democracy, equitable and fair outcomes are an important systemic goal.

\paragraph{Model of Fairness:}
In a participatory budgeting setting, there are $k$ projects (or items) and $n$ voters (or agents) who participate. Unlike in a private good economy, it is usually the case that $k \ll n$. There is an overall budget $B$ available for funding projects. An \textit{allocation} is a $k$-dimensional vector $\mathbf{x} \in \mathbb{R}^k$ with $\mathbf{x} \geq 0$ and $\sum_{j=1}^k x_j \leq B$. The quantity $x_j$ denotes the funding for project $j$. We  assume voters report a cardinal \textit{utility function}. We denote the utility of an agent $i$ given an allocation $\mathbf{x}$  as $U_i(\mathbf{x})$, and we assume this function is continuous, non-decreasing, and concave. 

In this model, we study {\em fair allocations}. In classical economic theory, a {\em fair} allocation is one that is {\em Pareto-efficient} and {\em envy-free}~\cite{varian}.  An allocation $\mathbf{x}$ is Pareto-efficient (or Pareto) if there is no other feasible allocation $\mathbf{y}$ such that $U_i(\mathbf{x}) \le U_i(\mathbf{y})$ for all voters $i$, with a strict inequality for at least one voter. This captures the notion that the allocation is not doing a disservice by under-allocating to all agents.  In the context of private goods, envy-freeness means that no agent prefers the allocation of another agent.  However, a different notion is needed for public goods, since the allocation is shared among all agents.  In this paper, the concept of fairness with which we work is the core. This notion is borrowed from cooperative game theory and was first phrased in game theoretic terms in~\cite{scarfCore}. It has been studied extensively even in public goods settings~\cite{lindahlCore,coreConjectureCounter}. Below, we define the core and an approximate notion of the core.

\begin{definition} An allocation $\mathbf{x}$ is a \textbf{core} solution if there is no subset $S$ of agents who, given a budget of $(|S|/n)B$, could compute an allocation $\mathbf{y}$ where every user in $S$ receives strictly more utility in $\mathbf{y}$ than $\mathbf{x}$, i.e., $\forall i \in S, \; U_i(\mathbf{y}) > U_i(\mathbf{x})$.
\label{definition:core}
\end{definition}

\begin{definition}
For $\alpha > 1$, an allocation $\mathbf{x}$ lies in the \textbf{$\alpha$-approximate} multiplicative (resp. additive) core if for any subset $S$ of agents, there is no allocation $\mathbf{y}$ using a budget of of $(|S|/n)B$, s.t. $U_i(\mathbf{y}) > \alpha U_i(\mathbf{x})$ (resp. $U_i(\mathbf{y}) >  U_i(\mathbf{x}) + \alpha$) for all $i \in S$.  
\label{definition:approximateCore}
\end{definition}

Note that when $S = \{1,2,\ldots, n\}$, the above constraints encode a weak version of Pareto-Efficiency.  Further, when $S$ is a singleton voter, the core captures {\em Sharing Incentive}, meaning that the voter gets at least as much utility as she would have obtained with budget $B/n$ dedicated to just her. In general, the core captures a {\em group sharing incentive}: No community of users suffers envy with respect to its share of the overall budget.

\paragraph{Some Clarifying Examples:}  We briefly consider some examples to clarify the concept of the core and compare it with other definitions of fairness. For simplicity in these examples, assume the utility function of the agents is linear, so $U_i(\mathbf{x}) = \sum_{j=1}^k u_{ij} x_j$.  Also, assume that there is a unit size budget and all projects are of unit size. 

\begin{figure}[h!]
\label{app:coreeg}
\centering
\caption{Core Examples}
\subfloat[Tyranny of the Majority]{
\resizebox{0.25\linewidth}{!}{
    \begin{tabular}{|c ||c c|} 
        \hline
        Agent $i$ & $u_{i,1}$ & $u_{i,2}$ \\ [0.5ex] 
        \hline\hline
        1 & 1 & 0 \\ 
        \hline
        2 & 1 & 0 \\
        \hline
        \vdots & \vdots & \vdots \\  
        \hline
        $\lceil (n/2) \rceil +1$ & 1 & 0 \\
        \hline
        $\lceil (n/2) \rceil +2$ & 0 & 1 \\ 
        \hline
        \vdots & \vdots & \vdots \\  
        $n$ & 0 & 1 \\ [1ex]
        \hline
    \end{tabular}
}
}\hspace{0.5cm}
\subfloat[Account for Sharing]{
\resizebox{0.25\linewidth}{!}{
    \begin{tabular}{|c ||c c c|} 
        \hline
        Agent $i$ & $u_{i,1}$ & $u_{i,2}$ & $u_{i,3}$ \\ [0.5ex] 
        \hline\hline
        1 & 3/5 & 0 & 2/5 \\ 
        \hline
        2 & 3/5 & 0 & 2/5 \\
        \hline
        \vdots & \vdots & \vdots & \vdots \\  
        \hline
        $\lceil (n/2) \rceil$ & 3/5 & 0 & 2/5 \\
        \hline
        $\lceil (n/2) \rceil + 1$ & 0 & 3/5 & 2/5 \\ 
        \hline
        \vdots & \vdots & \vdots & \vdots \\  
        $n$ & 0 & 3/5 & 2/5 \\ [1ex]
        \hline
    \end{tabular}
}
}\hspace{0.5cm}
\subfloat[Fairness and Proportionality]{
\resizebox{0.25\linewidth}{!}{
    \begin{tabular}{|c ||c c|} 
        \hline
        Agent $i$ & $u_{i,1}$ & $u_{i,2}$ \\ [0.5ex] 
        \hline\hline
        1 & 1 & 0 \\ 
        \hline
        2 & 1 & 0 \\
        \hline
        \vdots & \vdots & \vdots \\  
        \hline
        $n-1$ & 1 & 0 \\
        \hline
        $n$ & 0 & 1 \\ [1ex]
        \hline
    \end{tabular}
}
}
\label{CoreExamples}
\end{figure}

First note that the core will produce a very different outcome from simply aggregating ``yes/no" votes of voters for different projects, and funding the projects in decreasing order of votes.  If we run such a procedure on the example in Figure~\ref{CoreExamples}(a), the majority has one more vote than the minority, yet the majority is exclusively privileged. Figure~\ref{CoreExamples}(b)  demonstrates that the naive fair allocation to allow every agent to determine $1/n$ of the overall allocation is not Pareto-efficient; such a scheme would fund items 1 and 2 with half of the budget each even though all agents would be better off if some of the budget were spent on item 3. In Figure~\ref{CoreExamples}(c), we contrast a core solution with max-min fairness. An allocation is max-min fair if the utility of the agent with the least utility is maximized. In our example, this corresponds to funding each of the two items with half of the budget. While this is Pareto-efficient, it favors one voter at the expense of all the others. The core solution for the same example is to fund \textit{proportionally}: Items are funded in proportion to the number of voters preferring them.

\paragraph{High-level goals:}
At a high level, we explore three related questions in Sections~\ref{sec:core},~\ref{sec:eval}, and~\ref{sec:truthful} respectively:

\begin{itemize}
\item Can we efficiently compute core allocations for reasonably general utility functions? 
\item What do these allocations look like for data generated by real participatory budgeting instances under utility functions motivated by that data?
\item For simple utility functions, can we develop a truthful mechanism for computing core allocations without payments?
\end{itemize}

We positively answer the first and third question using techniques from optimization and differential privacy to develop the algorithmic understanding of the Lindahl equilibrium, a market based notion we will define shortly.  For the second question, we use our theoretical results to develop principled heuristics that we validate using data from the Stanford Participatory Budgeting Project~\cite{SPBP}.  Before proceeding however, we turn to consider utility functions more precisely.

 \subsection{Utility Functions}
 \label{sec:utility}
We consider utility functions generalizing the linear utility functions used in previous examples. These utility functions, which we term {\sc Scalar Separable}, have the form 
$$U_i(\mathbf{x}) = \sum_{j} u_{ij} f_j(x_j)$$
for every agent $i$ where $\{f_j\}$ are smooth, non-decreasing, and concave, and $\mathbf{u_{i}} \geq 0$.  By $\sum_j$ we always mean the sum over the $k$ projects.  {\sc Scalar Separable} utilities are fairly general and well-motivated. First, this concept encompasses linear utilities and several other canonical utility functions (see below). Secondly, if voters express scalar valued preferences (such as up/down approval voting), {\sc Scalar Separable} utilities provide a natural way of converting these votes into cardinal utility functions.  In fact, as we discuss below, we  will do precisely this when handling real data.  We consider two subclasses that we term {\sc Non-satiating} and {\sc Saturating} utilities respectively. Each arises naturally in settings related to participatory budgeting.

\paragraph{Non-satiating Functions:}
 For our main computational result in Section~\ref{sec:core}, we consider a subclass of utility functions that we term {\sc Non-satiating}.  
 \begin{definition}
\label{def:non-sat}
A differentiable, strictly increasing, concave function $f$ is {\sc Non-Satiating} if $x f_j'(x)$ is monotonically increasing and equal to $0$ when $x=0$.
\end{definition}
 
This is effectively a condition that the functions grow at least as fast as  $\ln{x}$.  Several utility functions used for modeling substitutes and complements fall in this class. For instance, constant elasticity of substitution (CES) utility functions where
$$ U_i(\mathbf{x}) = \left( \sum_j u_{ij} x_j^{\rho} \right)^{\frac{1}{\rho}} \ \ \mbox{ for } \rho \in (0,1]$$ 
can be monotonically transformed into {\sc Non-satiating} utilities\footnote{Note that the core remains unchanged if utilities undergo a monotone transform.}. CES functions are also {\em homogeneous of degree 1}, meaning that $U_i(\alpha \mathbf{x}) = \alpha U_i(\mathbf{x})$ for any scalar $\alpha \ge 0$.  The special case when $\rho = 1$ captures linear utilities. The case when $\rho \rightarrow 0$ captures {\sc Cobb-Douglas} utilities which can be written as 
$$U_i(\mathbf{x}) = \prod_j x_j^{\alpha_{ij}}$$
 where $\sum_j \alpha_{ij} = 1$ and $\alpha_{ij} > 0$.  We consider homogeneous functions of degree 1 in Section~\ref{sec:truthful} to design a randomized approximately truthful mechanism. 
 
%Consider voting schemes for public projects where the projects don't have budgets; instead, voters reveal their optimal monetary allocation.  In this case, if a single user's optimal allocation is to choose $x_{ij} = \alpha_{ij} B$, so that the vote of a user can be interpreted as their Cobb-Douglas coefficients.  

\paragraph{Saturating Functions:} Note that these utilities implicitly assume projects are divisible. Fractional allocations make sense in their own right in several scenarios: Budget allocations between goals such as defense and education at a state or national level are typically fractional, and so are allocations to improve utilities such as libraries, parks, gyms, roads, etc. However, in the settings for which we have real data, the projects are indivisible and have a monetary cost $s_j$, so that we have the additional constraint $x_j \in \{0,s_j\}$ on the allocations.  We describe such data from the Stanford Participatory Budgeting Platform~\cite{SPBP} in greater detail in Section~\ref{sec:eval}. We therefore need utility functions that model budgets in individual projects.  These utility functions must also be simple to account for the limited information elicited in practice.  For example, in the voting data that we use in our experiments, each voter receives an upper bound on how many projects she can select, and the ballot cast by a voter is simply the subset of projects she selects. A related voting scheme implemented in practice, called Knapsack Voting~\cite{goel2015knapsack}, has similar elicitation properties. For modeling these two considerations, we consider {\sc Saturating} utilities.

\begin{definition}
\label{eq:sat}
A utility function is in the {\sc Saturating} model if it has the form
\[
U_i(\mathbf{x}) = \sum_j u_{ij} \min{\left(x_j/s_j, 1 \right)}
\]
\end{definition}

For converting our voting data into a {\sc Saturating} utility, we set $s_j$ to be the budget of project $j$, and set $u_{ij}$ to $1$ if agent $i$ votes for project $j$ and $0$ otherwise. Note that if $x_j = s_j$, then the utility of any agent who voted for this item is $1$. This implies the total utility of an agent is the number (or total fraction) of items that he voted for that are present in the final allocation. Clearly, {\sc Saturating} utilities do not satisfy Def.~\ref{def:non-sat}.  However, we will connect {\sc Non-satiating} and {\sc Saturating} utilities  by developing an approximation algorithm and heuristic for computing core allocations in the saturating model using results developed for the {\sc Non-satiating} model.  

\subsection{Computing Core Solutions via the Lindahl Equilibrium}
In a fairly general public goods setting, there is a market based notion of fairness due to Lindahl~\cite{LindahlPaper} and Samuelson~\cite{SamuelsonPaper} termed the {\em Lindahl equilibrium}, which is based on setting different prices for the public goods for different agents.  The market on which the Lindahl equilibrium is defined is a mixed market of public and private goods.  We present a definition below that is specialized to just a public goods market relevant for participatory budgeting. 
 
\begin{definition}
\label{def:lindahl}
In a public goods market with budget $B$, per-voter prices $\mathbf{p_1}, \mathbf{p_2}, ..., \mathbf{p_n}$ each in $\mathbb{R^+}^k$ and allocation $\mathbf{x} \in \mathbb{R^+}^k$ constitute a \textbf{Lindahl equilibrium} if the following two conditions hold:
\begin{enumerate}
 \item For every agent $i$, the utility $U_i(\mathbf{y}_i)$ is maximized subject to $\mathbf{p_i} \cdot \mathbf{y}_i \leq B/n$ when $\mathbf{y}_i = \mathbf{x}$; and
\item The {\em profit} defined as $\left(\sum_i \mathbf{p_i} \right) \cdot \mathbf{z} - \|\mathbf{z}\|_1$, subject to $\mathbf{z} \ge 0$ is maximized when $\mathbf{z} = \mathbf{x}$.
\end{enumerate}
\end{definition}

The price vector for every agent is traditionally interpreted as a tax.  However, unlike in private goods markets, in our case these prices (or taxes) are purely hypothetical; we are only interested in the allocation that results at equilibrium (in fact, we eliminate the prices from our characterization of the equilibrium). Under innocuous conditions for the mixed public and private goods market, Foley proved that the Lindahl equilibrium exists and lies in the core~\cite{lindahlCore}.  This remains true in our specialized instance of the problem; the omitted proof is a trivial adaption from~\cite{lindahlCore}.  Thus, computing a Lindahl equilibrium is sufficient for the purpose of computing a core allocation.  However, Foley only proves existence of the equilibrium via a fixed point argument that does not lend itself to efficient computation.

 \subsection{Truthfulness and Mechanisms} 
In addition to investigating the computational complexity of the core, we also investigate dominant strategy truthfulness. We study {\em asymptotic approximate truthfulness}~\cite{LiuPycia2011}. For any agent $i$, reporting the true utility function $U_i(\mathbf{x})$ maximizes the expected utility of agent $i$, subject to all agents $i' \neq i$ reporting true utility functions $U_{i'}(\mathbf{x})$ and agent $i$ knowing these. Our notion is asymptotic in the sense that we assume $n \gg k$, which is reasonable in practice.  It will also be approximately truthful in the following sense.

\begin{definition}
For $\delta > 0$, a (randomized) allocation mechanism is $\delta-$ \textbf{approximately dominant-strategy truthful} if any agent's expected utility increases by at most an additive value of $\delta$ by misreporting, even when she knows the utilities of the other agents.
\label{definition:approximateTruthful}
\end{definition}

Though it is easy to deduce from previous work that the core always exists for participatory budgeting, it is also reasonably well-known that a core outcome is easy to manipulate. This is not especially surprising; truthfulness has long been considered a serious problem in economics for the allocation of public goods \cite{freeRiderProblem,coreConjectureCounter}. In particular, the core outcome suffers from the classic economic free rider problem: Many agents benefit from one common good and knowing that, another agent who also benefits from that item falsely reports that she gets no utility from that common item and instead gets high utility only from her own uniquely preferred items.
 
Consider the example in Figure~\ref{TruthfulnessExamples}(a) below showing utilities for $n$ agents and $2$ items. Again, for simplicity assume utilities are linear, so that $U_i(\mathbf{x}) = \sum_j u_{ij} x_j$, and assume the budget and all projects are of unit size.  Suppose agent 1 knows that every other agent wants the first item, which agent 1 also wants but likes less than the second item. Then there is incentive for agent 1 to lie and report her utility for item 1 as 0 and her utility for item 2 as 1, in which case she still benefits from a large allocation of the shared item, but benefits slightly more from the addition of more of her uniquely preferred item in the core solution. However, one may quickly point out that the additional amount of utility agent $1$ receives by lying goes to $0$ for a large value of $n$. This begs the question of whether the core outcome is truthful in the large market limit. %Note that similar results are known for market clearing solutions for private goods~\cite{budish}. 

\begin{figure}[h!]
\centering
\caption{Core Outcomes are Easy to Manipulate}
\subfloat[Free Rider]{
    \begin{tabular}{|c ||c c|} 
        \hline
        Agent $i$ & $u_{i,1}$ & $u_{i,2}$ \\ [0.5ex] 
        \hline\hline
        1 & 1/3 & 2/3 \\ 
        \hline
        2 & 1 & 0 \\
        \hline
        \vdots & \vdots & \vdots \\  
        \hline
        $n$ & 1 & 0 \\ [1ex]
        \hline
    \end{tabular}
}\hspace{0.5cm}
\subfloat[Large Market]{
    \begin{tabular}{|c ||c c|} 
        \hline
        Agent $i$ & $u_{i,1}$ & $u_{i,2}$ \\ [0.5ex] 
        \hline\hline
        1 & 1/3 & 2/3 \\ 
        \hline
        2 & 2/3 & 1/3 \\
        \hline
        3 & 1/2 & 1/2 \\
        \hline
        \vdots & \vdots & \vdots \\  
        \hline
        $n$ & 1/2 & 1/2 \\ [1ex]
        \hline
    \end{tabular}
}
\label{TruthfulnessExamples}
\end{figure}

Unfortunately, it is also well-known that this does not happen for public goods. In Figure~\ref{TruthfulnessExamples}(b), many agents are indifferent between two items and two agents each prefer one of the two. Because all but the first two agents are indifferent, the core solution will be entirely based on the reported benefits from the first two agents regardless of the total number of agents. Thus, the incentive to misreport does not necessarily vanish as $n$ becomes large.  

In these examples, the misreporting agent needs to very precisely know the preferences of the other agents. Indeed, this is necessary:  Computing the core of public goods with no private goods (\textit{e.g.}, money transfers) satisfies the property of \textit{strategy-proof in the large}~\cite{strategyProofInLarge},  meaning that if agents know a distribution from which the preferences of other agents are drawn, then the market is truthful in expectation in the limit as the number of agents grows large.  However, in this paper, we study the stronger notion of dominant strategy truthfulness.

\subsection{Our Results}
In Section~\ref{sec:core}, we present a simple characterization of the Lindahl equilibrium in terms of the allocation variables and a means of efficient computation for {\sc Non-satiating} utilities.  Together, this results in an efficient algorithm for computing the core exactly for {\sc Non-satiating} utilities via convex programming. As far as we are aware, this is the {\em first} non-trivial computational result for the Lindahl equilibrium.  

As a consequence of our characterization, if the utility functions are homogeneous of degree 1 and concave (or any monotone transform thereof), then the {\em proportionally fair} allocation, the extentsion of the Nash Bargaining solution~\cite{nashBargaining}) that maximizes $\sum_i \log U_i(\mathbf{x})$, computes the Lindahl equilibrium. This mirrors similar results for computing a Fisher equilibrium in private good markets~\cite{eisenbergGaleMarkets}. In addition, we show that for homogeneous functions, quadratic voting~\cite{qv} can be used to elicit the gradient of the proportional fairness objective, pointing to practical implementations in the field. For more general utility functions, our potential function can be viewed as a regularized version of the proportional fairness objective written on a non-linear transform of the utility function -- a result that is new to the best of our knowledge. We also note that the class of {\sc Non-satiating} utilities includes many functions that are not monotone transforms of homogeneous functions of degree 1, and for some of these functions, computing a Fisher equilibrium is intractable~\cite{marketHardness}.

In Section~\ref{sec:eval}, we consider the question of computing core solutions for real world data sets from the Stanford Participatory Budgeting Platform~\cite{SPBP} that we model using {\sc Saturating} utility functions as discussed in Section~\ref{sec:utility}. We present an approximation algorithm as well as a heuristic implementation inspired by our characterization.  On real data, we find that this heuristic efficiently compute the exact core. Surprisingly, the resulting outcomes match the welfare optimal solutions on the same utility functions, which shows that simple approval voting schemes produce fair outcomes in the field.

In Section~\ref{sec:truthful}, we address incentive concerns. Truthfulness has long been considered a serious problem for the allocation of public goods \cite{freeRiderProblem,coreConjectureCounter}. We study {\em asymptotic approximate truthfulness}~\cite{LiuPycia2011}. For any agent $i$, reporting the true utility function $U_i(\mathbf{x})$ maximizes the expected utility of agent $i$, subject to all agents $i' \neq i$ reporting true utility functions and agent $i$ knowing these. Our notion is asymptotic in the sense that  $n \gg k$, which is reasonable in practice.  
%We note there that computing the core of public goods with no private goods (\textit{e.g.}, money transfers) satisfies the property of \textit{strategy-proof in the large}~\cite{strategyProofInLarge},  meaning that if agents know a distribution from which the preferences of other agents are drawn, then the market is truthful in expectation in the limit as the number of agents grows large. 
 We show that when agents' utilities are linear (and more generally, homogeneous of degree 1), there is an efficient randomized mechanism that implements an $\epsilon$-approximate core solution as a dominant strategy for large $n$. We use the Exponential Mechanism~\cite{differentialPrivacy} from differential privacy to achieve this. The application of the Exponential Mechanism is not straightforward since the proportional fairness objective (that computes the Lindahl equilibrium) is not separable when used as a scoring function; the allocation variables are common to all agents. Furthermore, this objective varies widely when one agent misreports utility.   We define a scoring function directly based on the {\em gradient condition} of proportional fairness to circumvent this hurdle.

%---------------------------------------------------%
%---------------------------------------------------%
%---------------------------------------------------%
%---------------------------------------------------%
%---------------------------------------------------%

\subsection{Related Work}
%The core is a stability based solution concept. One of the original influential stability based solutions is the so-called Nash bargaining solution~\cite{nashBargaining}. 
%Nash studied the problem of solutions to two person nonzero-sum games wherein the players may iteratively bargain and noted that the only stable result maximized the product of utilities of the players. 
%This leads to the proportional fairness convex program for several market clearing problem, including the Eisenberg-Gale program for linear utility Fisher markets with private goods. 

The general literature characterizing private good market equilibrium and computation is extensive \cite{marketEquilibria,vaziraniBargaining,marketHardness,eisenbergGaleMarkets}; however, there is relatively little literature giving computational results for public good economies.  As discussed above, the proportional fairness algorithm, which has been extensively studied in private good markets~\cite{nashBargaining,eisenbergGaleMarkets}, need not find solutions in the core for {\sc Scalar Separable} utilities, and we can view our computational results as providing a non-trivial generalization of the proportional fairness concept to Lindahl equilibria. %We note that proportional fairness is widely implemented in practice, for example, in the TCP algorithm in networking~\cite{kellyTCP}, and in wireless networking~\cite{networkScheduling}.

In previous work~\cite{ROBUS}, a subset of authors studied randomized allocation of shared resources in a database system, and used a first principles proof to show that proportional fairness finds a core solution; this corresponds to the linear utility setting. As mentioned above, out computational results for {\sc Scalar Separable} utilities  are far more general and are based on characterizing the Lindahl equilibrium.
In participatory budgeting, neither linear utility functions nor randomized allocations are realistic assumptions. 

Our work is related to designing truthful mechanisms for combinatorial public projects~\cite{CPPP}. However, these works focused on the social welfare objective and utilized payments as does the well known VCG mechanism~\cite{vickrey,clarke}, which is impractical for the application of participatory budgeting.  It is well-known that proportional fairness is not implementable in a truthful mechanism for private goods. A strong negative result was shown recently: No strictly truthful mechanism can achieve more than a $1/2$-approximation to the proportionally fair solution~\cite{coleProportionalFairness}. Though these as well as public good markets are truthful in the Bayesian sense in the large market limit~\cite{strategyProofInLarge,budish2,budish3}, we seek dominant strategy truthful mechanisms, which are non-trivial to design for public good markets even in the large market limit. Finally, truthful allocation of private goods without money is classically referred to as {\em cake cutting}, on which there is extensive literature~\cite{cakeCutting}; but, there is scant work on allocating public goods truthfully without money.  The problem of truthful allocation of public goods without payments is considered in the context of the facility location problem in~\cite{DesignWithoutPayments}; however, the setting is unrelated to ours and the authors are concerned with the social welfare or the total dis-utility, not the core.

\subsection{Roadmap}
In Section~\ref{sec:core}, we present a characterization of the Lindahl equilibrium and means of efficient computation for {\sc Non-satiating} utilities. Using this, we present an approximation algorithm for {\sc Saturating} utilities in Section~\ref{sec:eval}, along with a heuristic implementation and evaluation on real-world data sets. We present the asymptotically truthful mechanism for linear utilities in Section~\ref{sec:truthful}.

\section{Non-Satiating Utilities: Characterization and Computation}
\label{sec:core}
Recall that in the participatory budgeting problem, there are $k$ items (or projects) and $n$ agents (or voters). It is typically the case that $k \ll n$. We will denote a generic voter by $i$ and a generic item by $j$. There is an overall budget of $B$. An \textit{allocation} is a $k$-dimensional vector $\mathbf{x} \in \mathbb{R}^k$ with $\mathbf{x} \geq 0$ and $\sum_{j=1}^k x_j \leq B$. We consider scalar separable utility, where the utility of an agent $i$ given an allocation $\mathbf{x}$ is denoted as $U_i(\mathbf{x}) = \sum_j u_{ij} f_j(x_j)$, where $\{f_j\}$ are smooth, non-decreasing, and concave, and $\mathbf{u_{i}} \geq 0$. 
\subsection{Characterization}

%From our discussion of envy-freeness and the core, we call an allocation {\em fair} if it lies in the core; recall the definition of core from Def.~\ref{definition:core}. 

Recall that in order to compute a core allocation it is sufficient to compute a Lindahl equilibrium (Definition~\ref{def:lindahl}). To do this, our first result is to develop a characterization of the Lindahl equilibrium that eliminates the price variables.

\begin{theorem}
An allocation $\mathbf{x} \geq 0$ corresponds to a Lindahl equilibrium if and only if 
\begin{equation}
\label{eq4}
\sum_i \left( \frac{u_{ij}f_{j}'(x_j)}{\sum_m u_{im} x_m f_{m}'(x_m) } \right) \leq \frac{n}{B}
\end{equation}
 for all items $j$, where this inequality is tight when $x_j > 0$.  
\label{theorem:eq}
\end{theorem}
\begin{proof}
We prove the statement for more general utility functions, $\{U_i(\mathbf{x})\}$. Recall the definition of Lindahl equilibrium from Definition~\ref{def:lindahl}: Condition (1) implies that there is a dual variable $\lambda_i$ for every agent such that 
\begin{equation}
\label{eq1}
\forall j \qquad \frac{\partial}{\partial x_j}U_i(\mathbf{x}) \le \lambda_i p_{ij}
\end{equation}
with the inequality being tight if $x_j > 0$. Condition (1) in Definition~\ref{def:lindahl} also implies that $\sum_j p_{ij} x_j = B/n$ for all $i$. Multiplying Equation (\ref{eq1}) by $x_j$, noting that the inequality is tight when $x_j > 0$, and summing,  
$$\forall i, \;\;\; \sum_j x_j \frac{\partial}{\partial x_j}U_i(\mathbf{x}) = \sum_j \lambda_i p_{ij} x_j = \lambda_i(B/n)$$
Rearranging
\begin{equation}
\label{eq2}
\forall i, \;\;\; \lambda_i = \frac{n}{B} \sum_j x_j \frac{\partial}{\partial x_j}U_i(\mathbf{x})
\end{equation}

Similarly, Condition (2) in Definition~\ref{def:lindahl} implies that $\forall x_j, \sum_i p_{ij} \leq 1$ where the inequality is tight when $x_j > 0$. Substituting into Inequality (\ref{eq1}) and summing, we have:
\begin{equation}
\forall x_j,\;\;\; \sum_i \frac{\frac{\partial}{\partial x_j}U_i(\mathbf{x})}{\lambda_i} \leq 1
\label{eq:lambda}
\end{equation}
with the inequality being tight when $x_j > 0$. Using Equation (\ref{eq2}) to eliminate $\lambda_i$ in Inequality (\ref{eq:lambda}), we finally obtain
$$\forall x_j,\;\;\; \sum_i \frac{\frac{\partial}{\partial x_j}U_i(\mathbf{x})}{\sum_m x_m \frac{\partial}{\partial x_m}U_i(\mathbf{x})} \leq \frac{n}{B}$$
with the inequality being tight when $x_j > 0$. Taking the appropriate partial derivatives in the {\sc scalar separable} utility model yields the theorem statement. 
\end{proof}

%In Appendix~\ref{sec:approx}, we prove that an additive approximation to the Lindahl equilibrium conditions implies an additively approximate core solution; the proof follows the outline of Lemma~\ref{lemma:LindahlCore}. Further, in Appendix~\ref{sec:concave}, we show that for {\sc separable} utilities, the characterization of the Lindahl equilibrium and thus the core to be that of a Nash equilibrium for a concave game over the items. Since concave games have at least one Nash equilibrium~\cite{Rosen}, this gives an alternative proof, again using fixed points, that Lindahl equilibria exist for such utilities.

%---------------------------------------------%
%---------------------------------------------%
%---------------------------------------------%

%---------------------------------------------%
%---------------------------------------------%
%---------------------------------------------%

\subsection{Efficient Computation}
We now present our main computational result that builds on the characterization above to give (to the best of our knowledge) the first polynomial time method for computing the Lindahl equilibrium.  We need the non-satiation assumption on the functions $\{f_j\}$ given in Def.~\ref{def:non-sat}.

\begin{theorem}
\label{the:potential}
When $U_i(\mathbf{x}) = \sum_j u_{ij} f_j(x_j)$ where $\{f_j\}$ satisfy Def.~\ref{def:non-sat}, the Lindahl equilibrium (and therefore a core solution) can be computed as the solution to a convex program.
\label{theorem:potential}
\end{theorem}
\begin{proof}
Theorem~\ref{theorem:eq} gives the characterization of the Lindahl equilibrium as
$$ \sum_i \left(\frac{ u_{ij}f_{j}'(x_j)}{\sum_m u_{im} x_m f_{m}'(x_m)} \right) \le \frac{n}{B} $$
for all items $j$, with the inequality tight if $x_j > 0$. 
\medskip

Define $z_j = x_j f_j'(x_j)$. Note that $x_j = 0$ iff $z_j = 0$. Since $f_j$ satisfies non-satiation, this function is continuous and monotonically increasing, and hence invertible.  Let $h_j$ be this inverse such that $h_j(z_j) = x_j$. Let $r_j(z_j) = h_j(z_j)/z_j = 1/f_j'(x_j)$. The Lindahl equilibrium characterization therefore simplifies to:
$$ \sum_i \left(\frac{ u_{ij}}{\sum_m u_{im} z_m} \right) \le \frac{n}{B} r_j(z_j) $$
with the inequality being tight when $z_j > 0$.  Let $R_j(z_j)$ be the indefinite integral of $r_j$ (with respect to $z_j$). Define the following potential function 

\begin{equation}
\label{eq:phi}
\Phi(\mathbf{z}) = \sum_i \log{\left( \sum_j u_{ij} z_j\right)} - \left( \frac{n}{B} \right) \sum_j R_j(z_j) 
\end{equation}

We claim that $\Phi(\mathbf{z})$ is concave in $\mathbf{z}$. The first term in the summation is trivially concave. Also, since $f_j'(x_j)$ is a decreasing function, $1/f_j'(x_j)$ is increasing in $x_j$. Since $r_j(z_j) = 1/f_j'(x_j)$, this is increasing in $x_j$ and hence in $z_j$. This implies $R_j(z_j)$ is convex, showing the second term in the summation is concave as well. It is easy to check that the optimality conditions of maximizing $\Phi(\mathbf{z})$ subject to $\mathbf{z} \ge 0$ are exactly the conditions for the Lindahl equilibrium. This shows that the Lindahl equilibrium corresponds to the solution to the convex program maximizing $\Phi(\mathbf{z})$. 
\end{proof}

Using Theorem~\ref{cor:approx} in Appendix~\ref{sec:approx}, an approximately optimal solution to the convex program gives an approximate core solution, which implies polynomial time computation for the convex program. We note that the non-satiation condition essentially implies that $f_j(x_j)$ should grow faster than $\ln{x_j}$.  In combination with the assumption that $f_j(x_j)$ is concave (i.e., that it grows no faster than linear in $x_j$), this leaves us with a broad class of concave functions for which the Lindahl equilibrium and hence the core can be efficiently computed.  

%--------------------------------------%
%--------------------------------------%

%--------------------------------------------------------%

\subsection{Connection to Proportional Fairness}
\label{sec:prop1}
The following is now a simple corollary of Theorem~\ref{the:potential}.

\begin{corollary}
\label{cor:prop}
If $U_i(\mathbf{x})$ is linear, {\em i.e.}, $U_i(\mathbf{x}) = \sum_j u_{ij} x_j$, or more generally, if it is homogeneous of degree 1, then the Lindahl equilibrium coincides with the proportionally fair allocation that maximizes $\sum_i \log U_i(\mathbf{x})$ subject to $\|\mathbf{x}\|_1 \le B$ and $\mathbf{x} \ge 0$.
\end{corollary}

The proof for the linear case is direct, and that for homogeneous functions uses a standard change of variables~\cite{marketEquilibria} and is omitted. As mentioned in Section~\ref{sec:utility}, an interesting special case of homogeneous functions concerns Cobb-Douglas utilities, where $U_i(\mathbf{x}) = \prod_j x_j^{\alpha_{ij}}$ where $\sum_j \alpha_{ij} = 1$ and $\alpha_{ij} > 0$. In this case, if a single agent could choose the whole allocation, the optimal choice would be $x_{j} = \alpha_{ij} B$. Suppose every agent i reveals these optimal allocations for themselves for every item $j$; call this $x_{ij}$. Then it is easy to check that the Lindahl equilibrium sets $x_j = \frac{1}{n} \sum_i x_{ij}$, which is simply the average of the individual monetary allocations.

\paragraph{Elicitation via Quadratic Voting.} For homogeneous functions, consider the proportional fairness program that maximizes $\sum_i \log U_i(\mathbf{x})$ subject to $\|\mathbf{x}\|_1 \le B$ and $\mathbf{x} \ge 0$. Let 
$$ F(\mathbf{x}) = \frac{1}{n} \sum_i \log U_i(\mathbf{x}) - \frac{1}{B} \|\mathbf{x}\|_1$$
It is easy to show that the proportional fairness program coincides with maximizing $F$. It is also easy to show that for such functions,
$$\frac{\partial F(\mathbf{x})}{\partial x_j} =  \frac{1}{n} \sum_i \frac{\frac{\partial}{\partial x_j}U_i(\mathbf{x})}{\sum_m x_m \frac{\partial}{\partial x_m}U_i(\mathbf{x})} -  \frac{1}{B}$$

Suppose users $i$ are drawn from some large population. Then one way to maximize $F$ is to perform stochastic gradient descent. Suppose the current point is $\mathbf{x_t}$. We sample a user $i$ at random, and for $\mathbf{x} = \mathbf{x_t}$,  estimate the quantity:
$$  \frac{\frac{\partial}{\partial x_j}U_i(\mathbf{x})}{\sum_m x_m \frac{\partial}{\partial x_m}U_i(\mathbf{x})} -  \frac{1}{B}$$
This will be an unbiased estimator of the gradient of $F$ at $\mathbf{x_t}$. 

Note now that the above expression only needs an estimate of the relative magnitudes of $\left\{ \frac{\partial}{\partial x_m}U_i(\mathbf{x}) \right\}$ at $\mathbf{x_t}$. In other words, it only needs an estimate of the {\em direction} of the gradient of $U_i(\mathbf{x})$ at $\mathbf{x} = \mathbf{x_t}$.  This in turn can be estimated by presenting user $i$ with an $\ell_2$-ball of radius $\epsilon$ around $\mathbf{x_t}$ and asking the user to maximize her utility, $U_i(\mathbf{x})$. In other words, the user solves the problem:
$$ \mbox{Maximize }\  U_i(\mathbf{x}) \ \ \mbox{s.t.} \ \ || \mathbf{x}  - \mathbf{x_t}||_2 \le \epsilon$$
As $\epsilon \rightarrow 0$,  simple calculus shows that the quantity $\mathbf{x}  - \mathbf{x_t}$ is in the direction of the gradient of  $U_i(\mathbf{x_t})$. This is termed {\em quadratic voting}~\cite{qv}, and gives a way to {\em elicit} enough information from individual voters in order to perform stochastic gradient descent and compute the proportionally fair allocation.

\paragraph{Beyond Proportional Fairness.} When the utility functions are not homogeneous, it is not clear how to express the potential function in Equation (\ref{eq:phi}) as running proportional fairness on a transformed space of allocations. For instance, if $U_i(\mathbf{x}) = \sum_j u_{ij} x_j^{\alpha_j}$, 
$$ \Phi(\mathbf{x}) = \sum_i \log \left(\sum_j \alpha_j u_{ij} x_j^{\alpha_j} \right) - \frac{n}{B} \sum_j \alpha_j x_j$$
This involves a non-linear transform of the utility function and a regularization term, which proportional fairness on any transformed input space does not capture. We also observe that running proportional fairness directly can be far away from the core. Consider an instance where agents are partitioned into groups $G_j$ where all agents in a group have non-zero utility for only item $j$, with utility function $u_{ij}f_j(x_j) = x_j^{\alpha_j}$ for some $\alpha_j \in (0,1)$. Since all groups have disjoint preferences, the core solution allocates $x_j$ in proportion to $|G_j|$. However, proportional fairness maximizes $\sum_{j} |G_j|\log{\left(x_j^{\alpha_j}\right)} = \sum_j \alpha_j |G_j| \log{x_j}$, which allocates $x_j$ in proportion to $\alpha_j |G_j|$.

\section{Saturating Utilities: Approximation and Experiments}
\label{sec:sat}
\label{sec:eval}
We now move to the question of modeling and analyzing real participatory budgeting data. We use data from seven different elections that used the Stanford Participatory Budgeting Platform (SPBP).  This platform (\url{http://pbstanford.org})~\cite{knapsack1,knapsack2} has been used by over 25 PB elections for digital voting. This platform incorporates multiple voting mechanisms including $K$-approval, knapsack, ranking, and comparisons. 

Voters are presented with a ballot containing descriptions of the candidate public projects with associated budgets. They are also presented with an overall budget. They can vote for at most a certain number of these projects, typically 4 or 5 (this voting method is called K-approval). Note that the projects chosen by a voter can exceed the total budget. The data set is therefore a $0/1$ matrix on projects and voters, where a $1$ denotes a vote by the voter for the  project.  The number of voters, $n$, ranges between 200 and 3000 in our datasets, and the number of items $k$ is at most 30.  A typical example is presented in Figure~\ref{EvaluationExample}.

For modeling such data, we need utility functions that respect the budget constraints of individual projects.  It is natural to use the {\sc Saturating} utility model (see Section~\ref{sec:utility}), where the utility of user $i$ is 
\begin{equation*}
U_i(\mathbf{x}) = \sum_j u_{ij} \min{\left(x_j/s_j, 1 \right)}
\end{equation*}
where $s_j$ is the budget of project $j$, and $u_{ij}$ is $1$ if $i$ votes for $j$ and $0$ otherwise. Therefore, the utility for $i$ if $j$ is chosen in the final allocation is $u_{ij} \in \{0,1\}$. Clearly, this function does not satisfy Def.~\ref{def:non-sat}. Nevertheless, we develop an approximation to the core that can be efficiently computed using a {\sc Non-satiating} relaxation of the utility model. Furthermore, we can show an even stronger result empirically: We can efficiently compute the {\em exact} core solutions under this utility model on our real-world data sets. We conclude this section with some observations on the relationship between welfare maximizing and core allocations in the saturating model.

\subsection{Efficiently Approximating the Core} 
\label{sec:saturation} 
Recall the definition of approximation from Def.~\ref{definition:approximateCore}.

\begin{theorem}
Given a collection of {\sc saturating} utility functions, let $s = \min_j s_j$. Then, for any $\epsilon > 0$, an $\alpha$-approximate multiplicative core can be efficiently computed, where $\alpha = (1/\epsilon) \left( B / s\right)^{\epsilon} + 1 - 1/\epsilon$. For $\epsilon = \log (B/s)$, we have a $O\left(\log \frac{B}{s} \right)$ approximation to the core.
\end{theorem}
\begin{proof}
Create the modified utility function $g_{j}(x_j)$ for item $j$ as:
\[g_{j}(x_j) =  \begin{cases} 
     \left( x_j / s_j\right) & x_j \leq s_j \\
      \left[ (1/\epsilon) \left( x_j / s_j\right)^{\epsilon} + 1 - 1/\epsilon \right]& x_j > s_j
   \end{cases}
\]
so that $\tilde{U}_i(\mathbf{x}) = \sum_j u_{ij} g_j(x_j)$. Clearly, $g_j$ is concave.  To see that $g_{j}$ satisfies Def.~\ref{def:non-sat}, observe that.
\[x_j g_{j}'(x_j) =  \begin{cases} 
      \left( x_j / s_j\right) & x_j \leq s_j \\
      \left( x_j / s_j\right)^{\epsilon}& x_j > s_j 
   \end{cases}
\]
From this, it is easy to see that $x_j g_{j}'(x_j)$ is increasing and equal to $0$ when $x_j = 0$, so Def.~\ref{def:non-sat} applies. By Theorem~\ref{theorem:potential},  we can efficiently compute a core solution.  Let $s = \min_{j} s_j$.  Then, for any agent $i$ and item $j$, the maximum utility that this modified function would compute is $u_{ij} \left[ (1/\epsilon) \left( B / s\right)^{\epsilon} + 1 - 1/\epsilon \right]$, whereas the true benefit would be $u_{ij}$, where $x_j \geq s_j$.  This gives an approximation factor of $\alpha = (1/\epsilon) \left( B / s\right)^{\epsilon} + 1 - 1/\epsilon$ in the utility functions. This easily implies the same approximation factor for the core. 
\end{proof}

%-----------------------------------%
%-----------------------------------%
%-----------------------------------%
%-----------------------------------%
%-----------------------------------%

\subsection{Heuristically Computing the Exact Core}
\label{sec:explore}
In our experiments, we show that an exact core solution can indeed be computed. This  crucially uses the characterization in Theorem~\ref{theorem:eq}.  Let $x_j \in [0,s_j]$ denote the current allocation to item $j$, and let $y_j = f'_j(x_j)$. The following complementarity condition relates $x_j$ and $y_j$:
$$\forall j,   \ \ y_j \le \frac{1}{s_j}  \qquad \mbox{and} \qquad x_j < s_j \ \ \Rightarrow \ \ y_j = \frac{1}{s_j} $$

The Lindahl equilibrium condition in Theorem~\ref{theorem:eq} can be written as:
$$\forall j, \ \  \frac{B}{n} \sum_i \frac{u_{ij} y_j}{\sum_m u_{im} x_m y_m} \le 1$$
with equality when $x_j > 0$. Given $\mathbf{x}_{-j}$ and $\mathbf{y}_{-j}$, we perform binary search on $x_j, y_j$ to satisfy the above non-linear equation subject to complementarity on $x_j, y_j$. We repeat this process, at each step choosing that item $j$ with the largest additive violation in the above inequality. We iterate until the Lindahl conditions for all items are satisfied to accuracy $\epsilon$. (e.g.,  $\epsilon  = 1/n$). By Theorem~\ref{cor:approx} in Appendix~\ref{sec:approx}, if this process converges, the result is an $\epsilon$-approximate additive core solution.

One issue is that these dynamics are not theoretically guaranteed to converge. Even empirically, there are instances where we observe cycling. To address this issue,  we perturb the vote matrix by small additive noise, so that  $u_{ij} \leftarrow u_{ij} + \mbox{Uniform}\left(0, \alpha\right)$, where $\alpha$ is a small constant like $1/k^2$. We empirically observe that the process now converges.  In Figure~\ref{fig:conv}, we show this behavior for three datasets with at least 2000 voters and 10 items each.  The convergence is comparable for all seven of our data sets; only three are shown for the sake of readability.

\begin{figure}
\centering
%\begin{minipage}{.45\textwidth}
  \centering
  \includegraphics[width=2.5in]{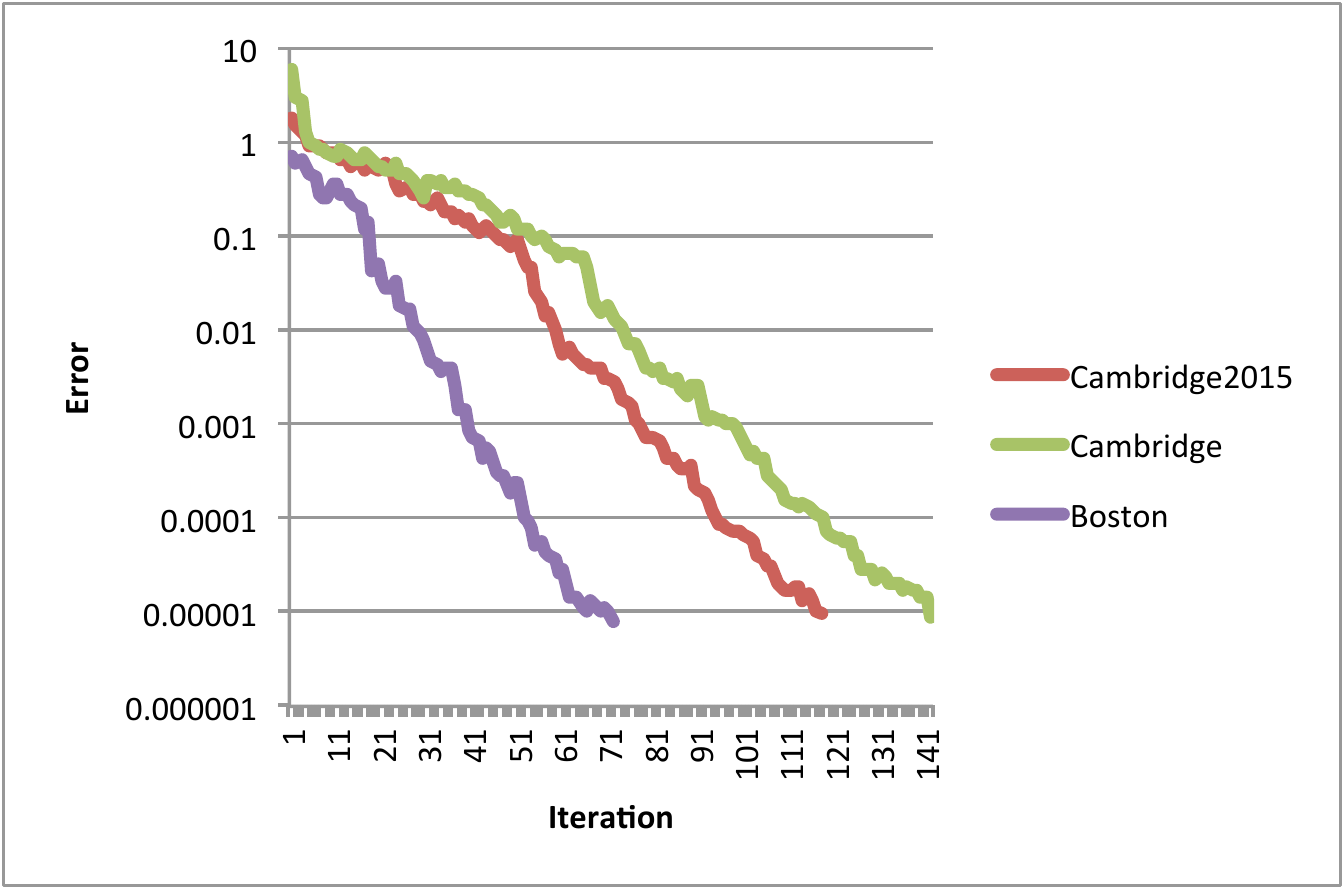}
  \captionof{figure}{Plot of error $\epsilon$ in Lindahl conditions as a function of number of iterations. }
  \label{fig:conv}
%\end{minipage}%
\end{figure}

\begin{observation}
Despite lacking a theoretical guarantee of convergence for {\sc Saturating} utilities, we are able to consistently compute  near-exact core solutions for our data sets using binary search on the complementarity conditions.
\end{observation}

%--------------------------------------------------------------------%
%--------------------------------------------------------------------%
%--------------------------------------------------------------------%

\subsection{Comparing the Core with {\sc Welfare}}
Given that we can compute the core exactly, we investigate its structure on our datasets.  We define the following vote aggregation schemes that we will use for comparison. The final allocation needs to be integral; we use heuristic methods to convert fractional allocations to integer ones. In the schemes below, the items are sorted in a certain order. Once sorted, for computing integer allocations, the schemes consider items in this order and add the item if its budget is less than the remaining total budget, stopping when all items are exhausted. For item $j$, let $n_j$ denote the number of votes received. Recall that these votes come from simple approval voting and that $s_j$ is the budget (size) of the item.  We can define fractional allocations similarly.  Importantly, both aggregation schemes use the same utility model.
\begin{itemize}
    \item {\sc Core}: Compute a fractional core allocation as described in Section~\ref{sec:explore}.  Let $x_j$ denote the fractional allocation of item $j$. Sort the items in descending of order $\frac{x_j}{s_j}$, which is the fraction to which item $j$ is funded in the fractional allocation.
        \item {\sc Welfare}: Sort the items in descending order of $\frac{n_j}{s_j}$. This  is the allocation that maximizes total (fractional) utility in the {\sc Saturating} utility model from Equation~\ref{eq:sat}.
\end{itemize}

\paragraph{Results.} We compare the outcomes of these algorithms for data sets from seven different real world instances of participatory budgeting.  We consider two measures of the similarities of outcomes: the Jaccard index and Budget similarity.  The Jaccard index for two integral allocations is the ratio of the size of their intersection to the size of their union. The Budget similarity for two fractional allocations $\mathbf{x}$ and $\mathbf{z}$ is defined as $\frac{\sum_j \min(x_j,z_j)}{B}$. Here, $\mathbf{x}$ is the actual monetary amount allocated to the project in the fractional allocation. 

Our composite results are shown in Figure~\ref{fig:similarity}. In Figure~\ref{EvaluationExample}, we show the results for Boston in detail as an example in which there are 10 projects, an overall budget of \$1,000,000 and over 2,000 voting agents.  Note that the integer allocations found by {\sc Core} and {\sc Welfare} are identical; in fact the {\sc Core} allocation is largely integral. 

\begin{figure}
    \centering
    \includegraphics[width=3in]{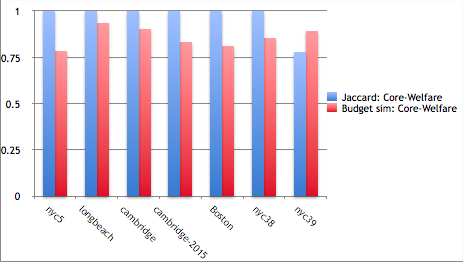}
    \captionof{figure}{Similarity Scores for {\sc Core} vs {\sc Welfare}. }
    \label{fig:similarity}
\end{figure}

\begin{figure}[h!]
\centering
{\small
    \begin{tabular}{| c | c | c | c | c |} 
    \hline
    Project & Budget & Votes & {\sc Core} & {\sc Welfare}  \\
    \hline \hline
        Wicked Free Wifi 2.0 & \$119,000 & 2,054 & 1.00 & 1.00 \\
    \hline
        Water Bottle Refill Stations at Parks & \$260,000 & 1,794 & 1.00 & 1.00\\ 
    \hline
        Hubway Extensions & \$101,600 & 737 & 1.00 & 1.00  \\
    \hline
    Bowdoin St. Roadway Resurfacing & \$100,000 & 611 & 1.00 & 1.00 \\
    \hline
        Bike Lane Installation & \$200,000 & 771 & 0.74 & 1.00  \\
    \hline
        Track at Walker Park & \$240,000 & 672 & 0.33 & 0.91  \\
    \hline
        BCYF HP Dance Studio Renovation & \$286,000 & 759 & 0.31 & 0.00  \\
    \hline
    BLA Gym Renovations & \$475,000 & 1044 & 0.20 & 0.00  \\
    \hline
    Ringer Park Renovation & \$280,000 & 546 & 0.02	& 0.00 \\
    \hline
    Green Renovation for BCYF Pino & \$250,000 & 452 & 0.01 & 0.00 \\
    \hline
    \end{tabular}}
\caption{Aggregation results for Boston. The Budget column lists the project's budget in dollars. The final two columns list the allocation of the project as a fraction of its budget, so that an integral allocation corresponds to 1.}
\label{EvaluationExample}
\end{figure}

%PAGE BREAK HERE TO AVOID SPLITTING UP THE OBSERVATION
\pagebreak

\begin{observation}
{\sc Core} and {\sc Welfare} compute the same integer allocations on almost all of our data sets, showing {\sc Welfare} produces fair allocations in practice. Furthermore, since the fractional allocation produced by {\sc Welfare} is an integer allocation except for one item, the high Budget similarity between {\sc Welfare} and {\sc Core} implies that the fractional core produces almost integer allocations.
\label{observation:welfare}
\end{observation}

The above observation that the {\sc Core} empirically coincides with {\sc Welfare} is quite surprising. It is easy to construct examples where the core allocation will be very different from welfare maximization. This is particularly pronounced when there is a significant minority of voters who have orthogonal preferences from the majority.  Therefore, one possible explanation for our observation is that users might have approximately independent random preferences over the projects. In Appendix~\ref{app:random}, we explore this possibility more formally.

\section{Homogeneous Utilities: Mechanism Design}
\label{sec:truthful}

\makeatletter
\newcommand{\distas}[1]{\mathbin{\overset{#1}{\kern\z@\sim}}}%
\newsavebox{\mybox}\newsavebox{\mysim}
\newcommand{\distras}[1]{%
  \savebox{\mybox}{\hbox{\kern3pt$\scriptstyle#1$\kern3pt}}%
  \savebox{\mysim}{\hbox{$\sim$}}%
  \mathbin{\overset{#1}{\kern\z@\resizebox{\wd\mybox}{\ht\mysim}{$\sim$}}}%
}

\label{sec;truthful}
In this section, we develop a randomized mechanism that finds an approximately core solution with high probability while ensuring approximate dominant-strategy truthfulness for all agents. In the spirit of~\cite{LiuPycia2011}, we assume the large market limit so that $n \gg k$; in particular, we assume $k = o(n^{1/2})$. We construct our mechanism for the special case of homogeneous utility functions of degree one. For simplicity, we present the mechanism for linear utility functions where $U_i(\mathbf{x}) = \sum_{j=1}^k u_{ij} x_j$, noting that it easily generalizes to degree one homogeneous functions. The values of $u_{ij}$ are reported by the agents. Without loss of generality, these are normalized so that $\|\mathbf{u_i}\|_1 = 1$.  Also without loss of generality, let $B$ be normalized to 1. Recall from Corollary~\ref{cor:prop} that for linear utility functions, the proportional fairness algorithm that maximizes $\sum_i \log U_i(\mathbf{x})$ subject to $\|\mathbf{x}\|_1 \le 1$ and $\mathbf{x} \geq 0$ computes the Lindahl equilibrium. 

\medskip
We will design additive approximations to the core (see Def.~\ref{definition:approximateCore}) that achieve approximate truthfulness in an additive sense (see Def.~\ref{definition:approximateTruthful}). We use the Exponential Mechanism~\cite{differentialPrivacy}  to achieve approximate truthfulness.  We will formally define the mechanism in Section~\ref{sec:exponentialMechanism}. At a high level, this mechanism is a general framework for implementing approximately truthful mechanisms for multi-agent optimization problems.  It arises from techniques in differential privacy where the goal is to minimize the sensitivity of database queries to the information of any given agent in the database.  This is accomplished by adding random noise to query responses, \textit{i.e.}, drawing the response from some distribution. For our application, this corresponds to drawing an allocation from a distribution rather than directly as the solution to an optimization problem; that distribution should be such that an individual agent is very unlikely to change their expected utility by misreporting their preferences.  The distribution is weighted exponentially according to some measure of quality; in our case the quality measure will correspond to approximating the core. However, the application of the Exponential Mechanism is not straightforward since the proportional fairness objective (that computes the Lindahl equilibrium) is not separable when used as a scoring function; the allocation variables are common to all agents. Furthermore, this objective varies widely when one agent misreports utility. We therefore need to define the scoring function carefully.

\subsection{The Scoring Function and its Approximation}
Fix a constant $\gamma \in (0,1)$ to be chosen later. We first define the convex set of feasible allocations as $\mathcal{P} := \{\mathbf{x} \,:\, \mathbf{x} \geq n^{-\gamma}, \|\mathbf{x}\|_1 \leq 1 \}$.  Note that all such allocations are restricted to allocating at least $n^{-\gamma}$ to each project.  Since the utility vector of any agent is normalized so $\|\mathbf{u_{i}}\|_1 = 1$, this implies that every agent gets a baseline utility of at least $n^{-\gamma}$, a fact we use frequently.  We define the following scoring function, which is based on the gradient optimality condition of Proportional Fairness: 
$$q(\mathbf{x}):= n - n^{-\gamma} \max_{\mathbf{y} \in \mathcal{P}} \left( \sum_i \frac{U_i(\mathbf{y})}{U_i(\mathbf{x})}\right)$$

We will approximately maximize this scoring function.  Aside from scaling for technical reasons, the idea is to minimize the quantity $\max_{\mathbf{y} \in \mathcal{P}} \left( \sum_i \frac{U_i(\mathbf{y})}{U_i(\mathbf{x})}\right)$.  This corresponds to finding points from the allocation space from which there is no direction that is preferable to many agents receiving little utility from the current point.  This ties naturally to the concept of a fair solution and can be seen as an interpretation of the gradient optimality condition of the proportional fairness program.  

The trade off in defining the scoring function is between reducing the sensitivity of the function to the report of an individual agent and thus improving the approximation to truthfulness, and having just enough sensitivity so that the mechanism defined in terms of the scoring function provides a good approximation to the core. Recall the $\alpha$-approximate additive core (definition~\ref{definition:approximateCore}).

\begin{theorem}
\label{theorem:core}
If $\max_{\mathbf{y} \in \mathcal{P}} \left( \sum_i \frac{U_i(\mathbf{y})}{U_i(\mathbf{x})}\right) = n$ then the allocation $\mathbf{x}$ is a $\frac{(k-1)n^{-\gamma}}{1-kn^{-\gamma}}$-approximate additive core solution.
\end{theorem}
\begin{proof}
Suppose by contradiction that $\max_{\mathbf{y} \in \mathcal{P}} \left( \sum_i \frac{U_i(\mathbf{y})}{U_i(\mathbf{x})}\right) = n$, but $\mathbf{x}$ is not a $\frac{(k-1)n^{-\gamma}}{1-kn^{-\gamma}}$-approximate core solution.  Then there is a subset $S$ of agents who want to deviate to some allocation $\mathbf{z}$ where $\|\mathbf{z}\|_1 \leq \frac{|S|}{n}$ and for every agent $i$ in $S$, $U_i(\mathbf{z}) > U_i(\mathbf{x}) + \frac{(k-1)n^{-\gamma}}{1-kn^{-\gamma}}$. Define the allocation %Consider the allocation $\mathbf{z'}$ created by scaling $\mathbf{z}$ to size $1-kn^{-\gamma}$ and then adding $n^{-\gamma}$ in every dimension.  More specifically, define:
$\mathbf{z'} := \frac{n\left(1-kn^{-\gamma}\right)}{|S|}\mathbf{z} + n^{-\gamma}\mathbf{1}$.
Clearly $\mathbf{z'} \in \mathcal{P}$, so by assumption we have that $\sum_i \frac{U_i(\mathbf{z'})}{U_i(\mathbf{x})}  \leq n $. Substituting for $\mathbf{z'}$ gives:
$$\sum_i \frac{\mathbf{u_i} \cdot \left( \frac{n\left(1-kn^{-\gamma}\right)}{|S|}\mathbf{z} + n^{-\gamma}\mathbf{1} \right)} {U_i(\mathbf{x})} \leq n$$
Thus, solving for $\sum_i \frac{\mathbf{u_i} \cdot \mathbf{z}}{U_i(\mathbf{x})}$ and simplifying:
\begin{equation*}
 \begin{split}
 \sum_i \frac{\mathbf{u_i} \cdot \mathbf{z}}{U_i(\mathbf{x})} & \leq \frac{|S|}{n\left(1-kn^{-\gamma}\right)}\left(n - \sum_i \frac{n^{-\gamma}}{U_i(\mathbf{x})}\right) \leq \frac{1-n^{-\gamma}}{1-kn^{-\gamma}}|S|
 \end{split}
\end{equation*}
However, recall that since $S$ is a deviating coalition from the approximate core, it should be that $U_i(\mathbf{z}) > U_i(\mathbf{x}) + \frac{(k-1)n^{-\gamma}}{1-kn^{-\gamma}}$ for all agents $i \in S$.  This implies that for all $i \in S$, it should be that $\frac{U_i(\mathbf{z})}{U_i(\mathbf{x})} > 1 + \frac{(k-1)n^{-\gamma}}{1-kn^{-\gamma}}$.  Thus, we can bound the sum over $i$ as $\sum_{i} \frac{U_i(\mathbf{z})}{U_i(\mathbf{x})} > \frac{1-n^{-\gamma}}{1-kn^{-\gamma}}|S|$.  This is a contradiction, completing the proof. %so $\mathbf{x}$ is a $\frac{(k-1)n^{-\gamma}}{1-kn^{-\gamma}}$-approximate core solution.
\end{proof}

Using essentially the same argument, the following corollary follows easily.

\begin{corollary}
\label{corollary:deltaCore}
If $\max_{\mathbf{y} \in \mathcal{P}} \left( \sum_i \frac{U_i(\mathbf{y})}{U_i(\mathbf{x})}\right) = n + \alpha$ then the allocation $\mathbf{x}$ is an $\frac{(k-1)n^{-\gamma} + \alpha n^{-1}}{1-kn^{-\gamma}}$-approximate additive core solution.
\end{corollary}
%--------------------------%
\begin{comment}
\begin{proof}
Suppose $\max_{\mathbf{y} \in \mathcal{P}} \left( \sum_i \frac{U_i(\mathbf{y})}{U_i(\mathbf{x})}\right) = n + \alpha$.  Then, by precisely the same argument as in Theorem~\ref{theorem:core}, we find that for deviating coalition $S$ and allocation $\mathbf{z}$ with $\|\mathbf{z}\|_1 \leq \frac{|S|}{n}$:
  $$\sum_{i \in S} \frac{\mathbf{u_i} \cdot \mathbf{z}}{U_i(\mathbf{x})} \leq \frac{1 - n^{-\gamma} + \alpha n^{-1}}{1-kn^{-\gamma}}|S|$$
If $\mathbf{z}$ is a deviating coalition from the $\mathbf{x}$ is an $\frac{(k-1)n^{-\gamma} + \alpha n^{-1}}{1-kn^{-\gamma}}$-approximate core then it must be that for all agents $i \in S$, $U_i(\mathbf{z}) > U_i(\mathbf{x}) + \frac{(k-1)n^{-\gamma} + \alpha n^{-1}}{1-kn^{-\gamma}}$.  This implies that
\begin{equation*}
\begin{split}
\sum_i \frac{\mathbf{u_i} \cdot \mathbf{z}}{U_i(\mathbf{x})} & > \left(1 + \frac{(k-1)n^{-\gamma} + \alpha n^{-1}}{1-kn^{-\gamma}} \right)|S| = \frac{1 - n^{-\gamma} + \alpha n^{-1}}{1-kn^{-\gamma}}|S|
\end{split}
\end{equation*}  
This is a contradiction, so $\mathbf{x}$ must be a $\frac{(k-1)n^{-\gamma} + \alpha n^{-1}}{1-kn^{-\gamma}}$-approximate core solution.
\end{proof}
\end{comment}
%--------------------------%
We now bound the sensitivity of $q(\mathbf{x})$ with respect to the report of one agent, as well as its range.  

\begin{comment}
\begin{claim}
$q(\mathbf{x}) \geq 0$
\end{claim}
\begin{proof}
\begin{equation*}
\begin{split}
q(\mathbf{x}) & = n - n^{-\gamma} \max_{\mathbf{y} \in \mathcal{P}} \left( \sum_i \frac{U_i(\mathbf{y})}{U_i(\mathbf{x})}\right) \\
 & \geq n - n^{-\gamma} \sum_i \frac{1}{n^{-\gamma}} \\
 & = n - n^{-\gamma} n^{1+\gamma} \\
 & = 0
\end{split}
\end{equation*}
\end{proof}
\end{comment}

\begin{lemma}
\label{claim:sensitivity}
$q(\mathbf{x}) \geq 0$ and $\max_{\mathbf{x} \in \mathcal{P}} q(\mathbf{x}) = n - n^{1-\gamma}$. Further, if $\Delta q$ is the largest possible difference in the scoring function between two sets of input differing only on the report of a single agent $i'$, ({\em i.e.}, the {\em sensitivity} of $q$), then $\Delta q = 1$.
\end{lemma}
\begin{proof}

To see the first part, note that $U_i(\mathbf{x}) \ge n^{-\gamma}$ for all $\mathbf{x} \in \mathcal{P}$. Therefore, 
$$ q(\mathbf{x})  = n - n^{-\gamma} \max_{\mathbf{y} \in \mathcal{P}} \left( \sum_i \frac{U_i(\mathbf{y})}{U_i(\mathbf{x})}\right) \ge n - n^{-\gamma} \sum_i \frac{1}{n^{-\gamma}} = 0$$
By the optimality condition of the Proportional Fairness convex program, 
$$ \min_{\mathbf{x} \in \mathcal{P}} \max_{\mathbf{y} \in \mathcal{P}} \left( \sum_i \frac{U_i(\mathbf{y})}{U_i(\mathbf{x})}\right) = n$$
Therefore, $\max_{\mathbf{x} \in \mathcal{P}} q(\mathbf{x}) = n - n^{1-\gamma}$. Similarly, when agent $i'$ misreports:
\begin{eqnarray*}
\Delta q & = & \left| \left( n - n^{-\gamma} \max_{\mathbf{y} \in \mathcal{P}} \left( \sum_i \frac{U_i(\mathbf{y})}{U_i(\mathbf{x})}\right)\right) - \left( n - n^{-\gamma} \max_{\mathbf{y} \in \mathcal{P}} \left( \frac{U_{i'}(\mathbf{y})}{U_{i'}(\mathbf{x})} + \sum_{i \neq i'} \frac{U_i(\mathbf{y})}{U_i(\mathbf{x})}\right)\right) \right| \\
& \leq & n^{-\gamma} \left( \max_{\mathbf{y} \in \mathcal{P}} \left( \frac{U_{i'}(\mathbf{y})}{U_{i'}(\mathbf{x})} \right) - 1 \right)  \leq 1
\end{eqnarray*}
The first inequality follows because we can assume w.l.o.g. that if the maximizing $\mathbf{y}$ changed for the misreported data, it yields a score no worse than the score of the original $\mathbf{y}$ on the misreported data, since that original $\mathbf{y}$ could have been chosen. 
\end{proof}

\subsection{Exponential Mechanism} 
\label{sec:exponentialMechanism}
We now plug the above scoring function into the Exponential mechanism from  \cite{differentialPrivacy}. We use $\epsilon > 0$ as the privacy approximation parameter, and thus as a parameter for the approximation of truthfulness.

\begin{definition}
Define $\mu$ to be a uniform probability distribution over all feasible allocations $\mathbf{x} \in \mathcal{P}$.  For a given set of utilities, let the mechanism $\zeta^{\epsilon}_q$ be given by the rule: $$\zeta^{\epsilon}_q:= \mbox{ choose } \mathbf{x} \mbox{ with probability proportional to } e^{\epsilon q(\mathbf{x})} \mu(\mathbf{x})$$
\end{definition}

The following lemma follows by using the sensitivity bound from Lemma~\ref{claim:sensitivity} in Theorem 6 from~\cite{differentialPrivacy}.

\begin{lemma}
\label{lemma:truthful}
$\zeta^{\epsilon}_q$ is $\left(e^{2\epsilon} - 1\right)$-approximately truthful.
\end{lemma}
The primary result of this section demonstrates that $\zeta^{\epsilon}_q$ can still find an approximate core solution while providing approximate truthfulness.
\begin{theorem}
\label{theorem:approximation}
If $k$ is $o(\sqrt{n})$ and $\frac{1}{\epsilon} > \frac{k n}{(n-k^2)\ln{n}}$ then $\zeta^{\epsilon}_q$ can be used to choose an allocation $\mathbf{x}$ that is an $O\left(\frac{k\ln{n}}{\epsilon\sqrt{n}}\right)$-approximate additive core solution w.p. $1-\frac{1}{n}$. 
\end{theorem}
\begin{proof}

Let $t = \frac{k+1}{ \epsilon} \ln{n}$.  Lemma 7 in \cite{differentialPrivacy} states that 
\begin{equation}\Pr\left[n - n^{-\gamma}\max_{\mathbf{y} \in \mathcal{P}} \left( \sum_i \frac{U_i(\mathbf{y})}{U_i(\mathbf{x})} \right) \leq OPT - 2t\right] \leq \frac{e^{-\epsilon t}}{\mu(S_t)} \label{ExponentialMechanismApproximation} \end{equation}  
where $OPT$ is the maximum value of $q(\mathbf{x})$ for feasible allocations $\mathbf{x}$ and $S_t = \{ \mathbf{x}: q(\mathbf{x}) > OPT - t\}$. By Lemma~\ref{claim:sensitivity}, we have $OPT = n(1-n^{-\gamma})$, but we need to bound $\mu(S_t)$, the probability that $\mathbf{x}$ drawn uniformly at random from $\mathcal{P}$ is in $S_t$.  We will show that $\mu(S_t) \geq n^{-k}$.  Let $\mathbf{x^*} \in \mathcal{P}$ be the allocation such that $q(\mathbf{x^*}) = OPT$. Since $\|\mathbf{x}\|_1 = 1$,  there is an item $j'$ with $x^*_{j'} \geq 1/k$. Let $\delta = 1/n$.  Define the set $S_{\delta}$ so that
\[ S_{\delta} = \left\{ \mathbf{x} \, : \, \begin{aligned} 
  &x^*_j \leq x_j \leq x^*_j + \delta  && j \neq j' \\
  &x^*_{j'} - k^2 \delta x^*_{j'} \leq x_{j'} \leq x^*_{j'} - k^2 \delta x^*_{j'} + \delta && j = j' 
\end{aligned} \right\} \]
It is not hard to see that since $x^*_{j'} \geq 1/k$, all $\mathbf{x} \in S_{\delta}$ are feasible.  Furthermore, because there is a ``width'' of $1/n$ in possible choice of $x_j$ for all $j$, $\mu(S_{\delta}) \geq n^{-k}$.  Thus, to complete the argument that $\mu(S_t) \geq n^{-k}$, we just need to show that $S_{\delta} \subseteq S_t$.  In our case, $$S_t = \left\{\mathbf{x} \,:\, \max_{y \in \mathcal{P}} \sum_i \frac{U_i(\mathbf{y})}{U_i(\mathbf{x})} < n + \frac{k+1}{\epsilon}\ln{n} \right\}$$   

Since $\frac{1}{\epsilon} > \frac{k n}{(n-k^2)\ln{n}}$, substituting shows that an allocation $\mathbf{x}$ is surely in $S_t$ if the same sum is less than $\frac{n^2}{n-k^2}$.  By construction, in the worst case for any agent $i$ and allocation $\mathbf{x} \in S_{\delta}$ (namely, if $u_{ij'} = 1$), $U_i(\mathbf{x}) \geq \frac{n-k^2}{n} U_i(\mathbf{x^*})$.  Therefore,  for all $\mathbf{x} \in S_{\delta}$ 
$$\max_{y \in \mathcal{P}} \sum_i \frac{U_i(\mathbf{y})}{U_i(\mathbf{x})} \leq \frac{n}{n-k^2} \max_{y \in \mathcal{P}} \sum_i \frac{U_i(\mathbf{y})}{U_i(\mathbf{x^*})} = \frac{n^2}{n-k^2}$$

Thus, we have that $S_{\delta} \subseteq S_t$ and therefore $\mu(S_t) \geq n^{-k}$.  Substituting into equation~\ref{ExponentialMechanismApproximation} and simplifying yields
$$\Pr\left[\max_{\mathbf{y} \in \mathcal{P}} \left( \sum_i \frac{U_i(\mathbf{y})}{U_i(\mathbf{x})} \right) > n + 2\frac{k+1}{ \epsilon} n^{\gamma}\ln{n}\right] \leq \frac{1}{n} $$

By applying Corollary~\ref{corollary:deltaCore}, we get that $\mathbf{x}$ chosen according to $\zeta^{\epsilon}_q$ is a $\frac{(k-1)n^{-\gamma} + 2(k+1)\epsilon^{-1} n^{\gamma-1}\ln{n}}{1-kn^{-\gamma}}$-approximate core solution with probability $1-\frac{1}{n}$.  Plugging in $\gamma = 1/2$ and using the fact that $k$ is $o(\sqrt{n})$ gives that $\mathbf{x}$ chosen according to $\zeta^{\epsilon}_q$ is an $O\left(\frac{k\ln{n}}{\epsilon \sqrt{n}}\right)$-approximate core solution with probability $1-\frac{1}{n}$.
\end{proof}

Finally, we show that $\zeta^{\epsilon}_q$ can be sampled in polynomial time~\cite{hitAndRun} with small additive error in truthfulness.

\begin{claim}
The Hit-and-run method can be used to sample according to $\zeta^{\epsilon}_q$ in polynomial time.
\label{theorem:log-concave}
\end{claim}
\begin{proof}
As argued in \cite{hitAndRun}, it is sufficient to show that $e^{\epsilon q(\mathbf{x})} \mu(\mathbf{x})$ is log-concave. The support of the sampling is clearly convex as it is just the feasible non negative orthant (feasibility defined by a hyperplane).  We need to show that the function is log-concave in the input $\mathbf{x}$.

$$\ln{\left(e^{\epsilon q(\mathbf{x})} \mu(\mathbf{x})\right)} = \epsilon q(\mathbf{x}) + \ln{\left(\mu(\mathbf{x})\right)}$$

Since $\mu$ is only uniform, this is just an affine transformation of $q(\mathbf{x})$, therefore we need only show that $q(\mathbf{x})$ is concave. Recall the definition of $q(\mathbf{x})$:

$$q(\mathbf{x}):= n -  n^{-\gamma} \max_{\mathbf{y} \in \mathcal{P}} \left( \sum_i \frac{U_i(\mathbf{y})}{U_i(\mathbf{x})}\right)$$

Each individual utility function $U_i$ is concave, since it is a linear function $(\mathbf{u_i} \cdot \mathbf{x})$.  Thus, $U(\mathbf{x})^{-1}$ is convex because it is the composition of the convex and non-increasing scalar function $1/x$ with the concave multivariate (but scalar valued) $U_i(x)$ \cite{boydConvexOptimization}.  The sum is still convex, as a linear combination of convex functions.  Then $q(\mathbf{x}) = n - n^{-\gamma} \max_{\mathbf{y} \in \mathcal{P}} \left( \sum_i \frac{U_i(\mathbf{y})}{U_i(\mathbf{x})}\right)$ is concave in $\mathbf{x}$.
\end{proof}

\section{Conclusion}
In this paper, we have initiated the computational study of the Lindahl equilibrium in order to address fair resource allocation in the context of participatory budgeting. Our key conceptual contribution is expressing the Lindahl equilibrium (and hence the core) purely in terms of the common allocation variables. In a sense, this is a mirror image of the role common prices play in private good markets. This allows us to efficiently compute core allocations as the solution to a convex program.  We also used our characterization to provide an adaptation of the exponential mechanism from differential privacy guaranteeing approximate truthfulness while computing approximate core allocations.  We studied the results that such core allocations produce on real data and saw that they are similar to welfare allocations under the saturating utility model.  We note that this is surprising as it is not obvious that core allocations should be similar to welfare allocations under \textit{any} utility model.

Our work is just the first step towards understanding participatory budgeting specifically and the fair allocation of public goods more generally. We do not yet understand the computational complexity for more general utility functions. Is computing the Lindahl equilibrium for public goods computationally hard or is there a polynomial time algorithm even without the non-satiating assumption? Our experimental results leave open intriguing questions about modeling of real voting data. In particular, is there a more formal explanation of why welfare appears fair in practice? Also, is there a different way to elicit more information from voters for a more precise modeling of their utility than just approval voting?

%%% Local Variables: 
%%% mode: latex
%%% TeX-master: "main"
%%% End: 

\paragraph{Acknowledgement.} We thank Anilesh Krishnaswamy for useful discussions, and the Stanford Crowdsourced Democracy Team for the use of their data.

\newpage
\bibliographystyle{plain}
\bibliography{refs}

\section*{Appendix}
\appendix

\section{Approximate Lindahl Equilibrium}
\label{sec:approx}
We prove that an additive approximation to the Lindahl equilibrium conditions implies an additively approximate core solution. 

\begin{theorem}
\label{cor:approx}
For any $\epsilon > 0$, suppose there is an allocation $\mathbf{x}$ such that for all items $j$, $x_j >0$ implies
$$\left | \frac{B}{n} \sum_i \left( \frac{\frac{\partial}{\partial x_j}U_i(\mathbf{x})}{\sum_m x_m \frac{\partial}{\partial x_m}U_i(\mathbf{x})} \right) - 1 \right | \le \epsilon $$
and $x_j = 0$ implies
$$ \frac{B}{n} \sum_i \left( \frac{\frac{\partial}{\partial x_j}U_i(\mathbf{x})}{\sum_m x_m \frac{\partial}{\partial x_m}U_i(\mathbf{x})} \right) \le 1+\epsilon $$
then $\mathbf{x}$ is an approximate core solution in the following sense:
\begin{enumerate}
\item $\sum_j x_j \le B/(1-\epsilon)$, and
\item For any subset $S$ of agents, there is no allocation $\mathbf{y}$ of size $\left(\frac{|S|}{n} - \epsilon \right) B$ such that $U_i(\mathbf{y}) > U_i(\mathbf{x})$ for all $i \in S$. 
\end{enumerate}
\end{theorem}
\begin{proof}
The first part is straightforward.  Define the following vector of prices:
$$p_{ij} = \frac{B}{n} \left( \frac{\frac{\partial}{\partial x_j}U_i(\mathbf{x})}{\sum_m x_m \frac{\partial}{\partial x_m}U_i(\mathbf{x})} \right)$$
For this price vector, $\sum_j p_{ij} x_j = \frac{B}{n}$. This price vector also satisfies $\frac{\partial}{\partial x_j}U_i(\mathbf{x}) p_{im} = \frac{\partial}{\partial x_m}U_i(\mathbf{x}) p_{ij}$ for all $j \neq m$. This implies the allocation $\mathbf{x}$ maximizes $U_i(\mathbf{z})$ subject to $\sum_j p_{ij} z_j \le \frac{B}{n}$. Let $\sum_i p_{ij}  = 1 + \alpha_j$. We have $\alpha_j \le \epsilon$ for all $j$.  Consider the profit function $P(\mathbf{z}) = \sum_i \mathbf{p_i} \cdot \mathbf{z} - \sum_j (1 + \alpha_j) z_j$. This function is identically 0 for all $\mathbf{z}$.

Now, suppose by contradiction that there exists $S$ and allocation $\mathbf{y}$ of size $\left(\frac{|S|}{n} - \epsilon \right) B$ such that $U_i(\mathbf{y}) > U_i(\mathbf{x})$ for all $i \in S$. This implies $\sum_j p_{ij} y_j > \frac{B}{n}$ for all $i \in S$. Summing, 
$$ \sum_{i \in S} \mathbf{p_i} \cdot \mathbf{y} - \frac{|S|}{n} B > 0 \ \ \Rightarrow \ \ \sum_{i \in S} \mathbf{p_i} \cdot \mathbf{y} - \sum_j y_j > \epsilon B$$
Note now that $\sum_j \alpha_j y_j \le \epsilon \sum_j y_j \le \epsilon B$. This implies
$$   \sum_{i \in S} \mathbf{p_i} \cdot \mathbf{y} - \sum_j (1+\alpha_j) y_j  > 0$$
This is a contradiction, since $P(\mathbf{y})=0$. 
\end{proof}

\section{Core under Independent Preferences}
\label{app:random}

Recall the intuition from Section~\ref{sec:sat} that one possible explanation for the similarity between core and welfare outcomes is that users might have approximately independent random preferences over the projects.  Consider a random model in which there are infinitely many agents and there is a value $p_j \in [0,1]$ associated with every project, all of which have unit cost.  Each agent votes for project $j$ with probability $p_j$, and these draws are independent across the projects.  Every project also has a utility $u_j$ associated with it so that the utility of an agent is the sum over the projects for which the agent votes of their $u_j$. The allocation maximizing {\sc Welfare} is then just the set $S^* = \mbox{argmax}_{S : |S| \le B} \sum_{j \in S} p_j u_j $.

For the theorem below, we consider a $(1+\epsilon)$- approximate {\em integral} core, where there is no subset of agents of size $a \cdot B$ who can deviate and choose $a \cdot B$ items integrally so that all agents improve their utility by at least a factor of $(1+\epsilon)$. For subset $S$, let $U(S)$ be the random variable denoting utility than an agent derives from $S$. Note that $\mathbb{E}[U(S)] = \sum_{j \in S} u_j p_j$, and {\sc Welfare} generates expected utility $\mathbb{E}[U(S^*)]$.

\begin{theorem} 
\label{thm:random}
Under the random users model, if $\mathbb{E}[U(S^*)] > \frac{1}{\epsilon} \sqrt{B \ln{(B)}}$ then $S^*$ is a $(1+\epsilon)$-approximate integral core solution.
\end{theorem}
\begin{proof}
Suppose by contradiction that $S^*$ is not a $1+\epsilon$ approximate core solution.  Then there must exist some $\alpha$ fraction of the agents who want to deviate to another allocation: call this set of items $S$, where $|S| \leq \alpha B$.  Then it must be that the probability of an agent preferring $S$ is at least $\alpha$.  Also, they must prefer it even subject to a $1+\epsilon$ multiplicative penalty, that is, $U(S) > (1+\epsilon)U(S^*)$.  

Let $\mathcal{S} = (1+\epsilon)U(S^*) - U(S)$.  $\mathcal{S}$ is the sum of at least $B$ random variables, and it's expectation is at least $\epsilon \mathbb{E}[U(S^*)]$ since $\mathbb{E}[U(S)] \le \mathbb{E}[U(S^*)]$.  We apply Hoeffding's inequality to get:
$$\Pr[\mathcal{S} < 0] \leq e^{-2(\epsilon \mathbb{E}[U(S^*)])^2/B}$$
However, recall that $\mathbb{E}[U(S^*)] > \frac{1}{\epsilon} \sqrt{B \ln{(B)}}$ by assumption, so $Pr[\mathcal{S} < 0] \leq 1/B^2$.  So, the probability of an agent preferring $S$ is no more than $1/B^2$.  But note that $\alpha$ must be at least $1/B$ in order for $S$ to be nonempty.  This is a contradiction, and $S^*$ is a $(1+\epsilon)$- approximate integral core solution. 
\end{proof}

This leads us to empirically test the independent preference hypothesis on our data sets. For each pair of items $j$ and $j'$, we perform a $\chi^2$-test of independence between the preference vectors for these items. Since preferences are binary, this test has two degrees of freedom. This produces a $p$-value; we mark the items as correlated if the $p$-value is less than $0.1$, and mark them as independent otherwise. We set the distance between two projects to $0$ if they are correlated and $1$ if they are independent, and run average linkage clustering on the resulting distance matrix. The results for the Boston data is presented in Figure~\ref{fig:clustering}; other data sets produce similar results. We observe that there are large groups of projects all of which are correlated with each other, as one might expect. This shows that the independent preference model is not the complete explanation for why {\sc Core} coincides with {\sc Welfare} on our data sets. %and we leave modeling our data and finding the right explanation as an open question.

\begin{figure}[h!]
    \centering
    \includegraphics[width=3.5in]{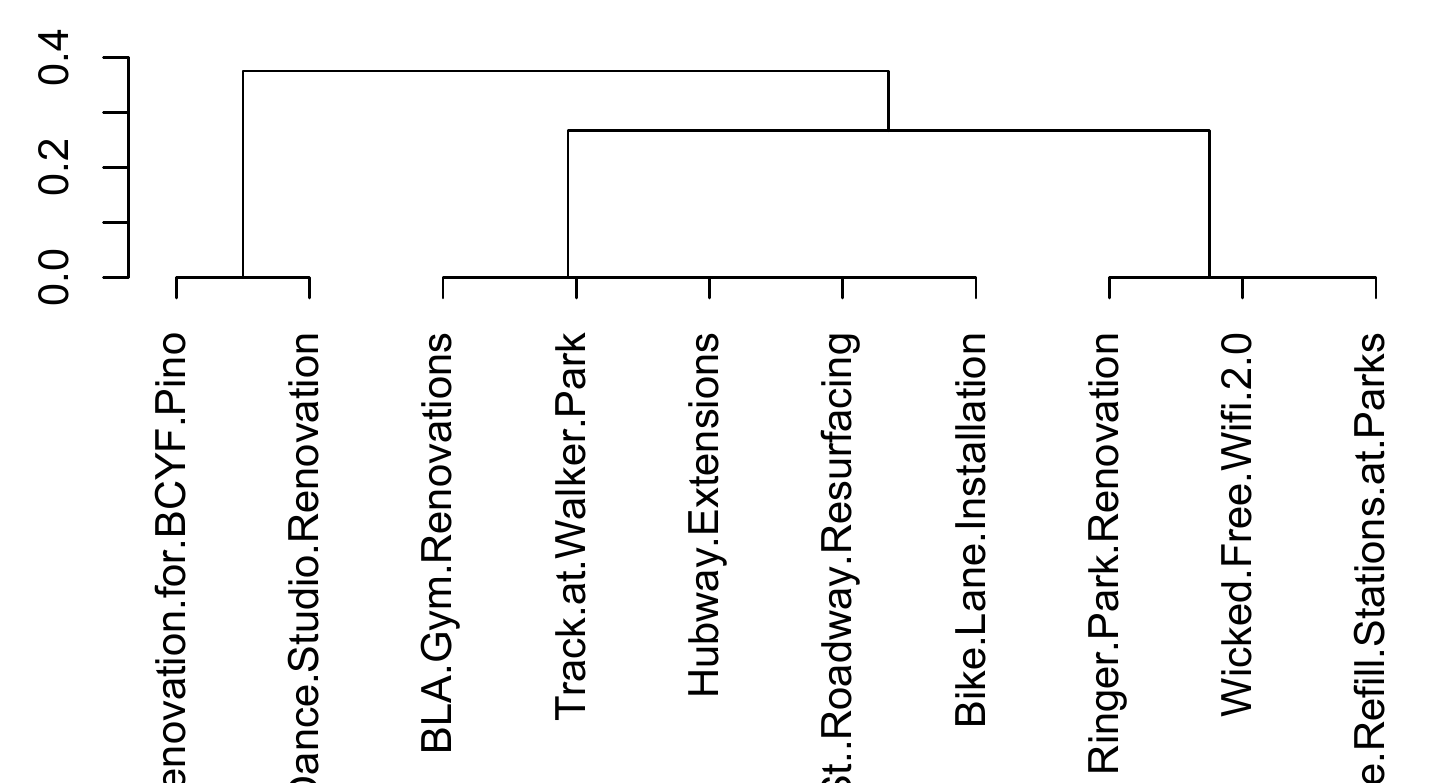}
    \captionof{figure}{Average linkage clustering dendrogram for items in the Boston data. A height of 0 denotes correlation and 1 denotes independence.  }
    \label{fig:clustering}
\end{figure}

\end{document}